\RequirePackage{amsthm}
\documentclass{ecai2014}

\usepackage{tikz}
\usetikzlibrary{automata,positioning}
\usepackage{times}
\usepackage{helvet}
\usepackage{courier}
\usepackage{mathtools}
\usepackage{shadow}
\usepackage{gastex}
\usepackage{xspace}
\usepackage{enumerate}
\usepackage{pbox}
\usepackage{array}
\usepackage{graphicx}
\usepackage{multirow}
\usepackage{changebar}
\usepackage{color}
\usepackage[skip=0pt]{caption}

\newtheorem{definition}{Definition}
\newtheorem{lemma}{Lemma}

\newtheorem{example}{Example}

\def\tool{\textit{LfSat}}

\def\dnf#1{\mathit{NF}(#1)}

\def\polsat{\emph{Polsat}}

\def\of#1{\mathit{of}(#1)}
\def\off#1{\mathit{off}(#1)}

\def\ofg#1{\mathit{ofg}(#1)}
\def\ofr#1{\mathit{ofr}(#1)}
\def\ofi#1{\mathit{ofx}(#1)}

\def\ltlf{\textit{LTL}_f}
\def\ff{\mathsf{ff}}
\def\tt{\mathsf{tt}}

\def\expand{\mathit{NF}}

\newcommand{\tran}[1]{\xrightarrow[]{#1}}

\begin{document}
%

\title{$\ltlf$ Satisfiability Checking}
\author{Jianwen Li
\institute{East China Normal University, China}
\and Lijun Zhang
\institute{Institute of Software, Chinese Academy of Sciences, China}
\and Geguang Pu$^1$
\and Moshe Y. Vardi
\institute{Rice University, USA}
\and Jifeng He$^1$
 }

\maketitle
\bibliographystyle{ecai2014}

\begin{abstract}
We consider here Linear Temporal Logic (LTL) formulas interpreted
over \emph{finite} traces. We denote this logic by $\ltlf$. The existing
approach for $\ltlf$ satisfiability checking is  based on a reduction to
standard LTL  satisfiability checking. We describe here a novel direct
approach to $\ltlf$ satisfiability checking, where we take advantage of the
difference in the semantics between LTL and $\ltlf$. While LTL
satisfiability checking requires finding a \emph{fair cycle} in an
appropriate transition system, here we need to  search only for a finite
trace. This enables us to introduce specialized  heuristics, where we also
exploit recent progress in Boolean SAT solving. We have implemented our
approach in a prototype tool and experiments show that our approach
outperforms existing approaches.
\end{abstract}

\section{Introduction}
Linear Temporal Logic (LTL) was first introduced into computer science
as a property language for the verification for non-terminating reactive
systems \cite{Pnu77}. Following that, many researches in AI have been
attracted by LTL's rich expressiveness. Examples of applications of LTL in AI
include temporally extended goals in planning \cite{BK98,DV99,CDV02,PLGG11},
plan constraints \cite{BK00,Gab04}, and user preferences \cite{BFM06,BFM11,SBM11}.

In a recent paper \cite{GV13}, De Giacomo and Vardi argued that while
standard LTL is interpreted over \emph{infinite} traces, cf. \cite{Pnu77},
AI applications are typically interested only in \emph{finite} traces.
For example, temporally extended goals are viewed as finite desirable
sequences of states and a plan is correct if its execution succeeds in
yielding one of these desirable sequences. Also in the area of business-process
modeling, temporal specifications for declarative workflows are interpreted
over finite traces \cite{APS09}. De Giacomo and Vardi, therefore, introduced
$\ltlf$, which has the same syntax as LTL but is interpreted over finite
traces.

In the formal-verification community there is by now a rich body of
knowledge regarding automated-reasoning support for LTL. On one hand, there
are solid theoretical foundations, cf. \cite{Var96}. On the other hand,
mature software tools have been developed, such as SPOT \cite{DP04}.
Extensive research has been conducted
to evaluate these tools, cf. \cite{RV10}. While the basic theory for $\ltlf$
was presented at \cite{GV13}, no tool has yet to be developed for $\ltlf$,
to the best of our knowledge. Our goal in this paper is to address this gap.

Our main focus here is on the \emph{satisfiability problem}, which asks
if a given formula has satisfying model.  This most basic automated-reasoning
problem has attracted a fair amount of attention for LTL over the past few years
as a principled approach to \emph{property assurance}, which seeks to eliminate
errors when writing LTL properties, cf. \cite{RV10,LZPVH13}. 

De Giacomo and Vardi studied the computational complexity of $\ltlf$ satisfiability
and showed that it is PSPACE-complete, which is the same complexity as for
LTL satisfiability \cite{SC85}. Their proof of the upper bound uses
a reduction of $\ltlf$ satisfiability to LTL satisfiability. That is,
for an $\ltlf$ formula $\phi$, one can create an LTL formula $\phi'$
such that $\phi$ is satisfiable iff $\phi'$ is satisfiable; furthermore,
the translation from $\phi$ to $\phi'$ involves only a linear blow-up. The
reduction to LTL satisfiability problem can, therefore, take advantage of existing
LTL satisfiability solvers \cite{SD11,LZPVH13}.
 On the other hand, LTL satisfiability checking requires reasoning about
infinite traces, which is quite nontrivial algorithmically, cf. \cite{CVWY92},
due to the required fair-cycle test.  Such reasoning is not required for $\ltlf$
satisfiability. A reduction to LTL satisfiability, therefore, may add unnecessary
overhead to $\ltlf$ satisfiability checking.

This paper approaches the $\ltlf$ satisfiability problem directly.
We develop a direct, and more efficient, algorithm for checking
satisfiability of $\ltlf$, leveraging the existing body of knowledge
concerning LTL satisfiability checking. The finite-trace semantics
for $\ltlf$ is fully exploited, leading to considerable simplification
of the decision procedure and significant performance boost. The finite-trace
semantics also enables several heuristics that are not applicable to LTL
satisfiability checking. We also leverage the power of advanced Boolean SAT
solvers in our decision procedure. We have implemented the new approach
and experiments show that this approach significantly outperforms the
reduction to LTL satisfiability problems.

The paper is organized as follows. we first introduce the definition
of $\ltlf$, the satisfiability problem, and the associated transition
system in Section~\ref{sec:pre}. We then propose a direct satisfiability-checking
framework in Section \ref{sec:checking}. We discuss various optimization
strategies in Section \ref{sec:optimizations}, and present experimental
results in Section \ref{sec:exp}. Section \ref{sec:con} concludes the paper.

\section{Preliminaries}\label{sec:pre}

\subsection{LTL over Finite Traces}
The logic $\ltlf$ is a variant of LTL.  Classical LTL formulas are
interpreted on infinite traces, whereas $\ltlf$ formulas are defined
over the finite traces.
Given a set $\mathcal{P}$ of atomic propositions, an $\ltlf$ formula
$\phi$ has the form:

$\phi ::= p\ |\ \neg \phi\ |\ \phi\vee\phi\ |\ \phi\wedge\phi\
|\ X\phi\ |\ X_w\phi\ |\ \phi U\phi\ |\ \phi R\phi$

where $X$(strong Next), $X_w$(weak Next), $U$(Until), and $R$(Release) are
temporal operators. We have $X_w\phi \equiv \neg X\neg \phi$ and
$\phi_1 R\phi_2\equiv \neg(\neg\phi_1 U\neg\phi_2)$. Note that in $\ltlf$,
$X\phi\equiv X_w\phi$ is not true, which is the case in LTL.

For an atom $a\in\mathcal{P}$, we call it or its negation ($\neg a$)
a literal. We use the set $L$ to denote the set of literals, i.e.
$L=\mathcal{P}\cup \{\neg a| a\in\mathcal{P}\}$.
Other boolean operators, such as $\rightarrow$
and $\leftrightarrow$, can be represented by the combination ($\neg,\vee$)
or ($\neg,\wedge$), respectively, and we denote the
constant \textit{true} as $\tt$ and \textit{false} as $\ff$.
Moreover, we use the notations $G\phi$ (Global) and $F\phi$
(Eventually) to represent $\ff R\phi$ and $\tt U\phi$. We use $\phi,\psi$
to represent $\ltlf$ or LTL formulas, and $\alpha,\beta$ for
propositional formulas.

Note that standard $\ltlf$ has the same syntax as LTL, see \cite{GV13}.
Here, however, we introduce the $X_w$ operator, as we consider $\ltlf$
formulas in NNF (Negation Normal Form), which requires all negations
to be pushed all the way down to atoms.  So a dual operator for $X$
is necessary.  In LTL the dual of $X$ is $X$ itself, while in $\ltlf$ it is
$X_w$.

\noindent
{\bf Proviso}:
In the rest of paper we assume that all formulas
(both LTL and $\ltlf$) are in NNF, and thus there are types of
formulas, based on the primary connective:
$\tt$, $\ff$, literal, $\wedge$, $\vee$, $X$
(and $X_w$ in $\ltlf$), $U$ and $R$.

The semantics of $\ltlf$ formulas is interpreted over finite traces,
which is referred to as the $\ltlf$ interpretations \cite{GV13}. Given
an atom set $\mathcal{P}$, we define $\Sigma := 2^{L}$. Let
$\eta\in\Sigma^*$ with $\eta=\omega_0\omega_1\ldots\omega_n$, we use
$|\eta|=n+1$ to denote the length of $\eta$. Moreover, for $1\leq
i\leq n$, we use the notation $\eta^i$ to represent
$\omega_0\omega_1\ldots\omega_{i-1}$, which is the prefix of $\eta$
before position $i$ ($i$ is not included).  Similarly, we also use
$\eta_i$ to represent $\omega_i\omega_{i+1}\ldots\omega_n$, which is
the suffix of $\eta$ from position $i$. Then we define $\eta$ models
$\phi$, i.e. $\eta\models\phi$ in the following way:

\begin{itemize}
  \item $\eta\models\tt$ and $\eta\not\models\ff$;
  \item If $\phi=p$ is a literal, then $\eta\models\phi$ iff
  $p\in\eta^1$;
  \item If $\phi=X\psi$, then $\eta\models\phi$
  iff $|\eta|>1$ and $\eta_1\models\psi$;
  \item If $\phi=X_w\psi$, then $\eta\models\phi$
  iff $|\eta|>1$ and $\eta_1\models\psi$, or $|\eta|=1$;
  \item If $\phi=\phi_1 U\phi_2$ is an Until formula, then
  $\eta\models\phi$ iff there exists $0\leq i< |\eta|$
  such that $\eta_i\models\phi_2$, and for every $0\leq j < i$ it
  holds $\eta_j\models\phi_1$ as well;
  \item If $\phi=\phi_1 R\phi_2$ is a Release formula, then
  $\eta\models\phi$ iff either for every $0\leq i< |\eta|$
  $\eta_i\models\phi_2$ holds, or there exists
  $0\leq i< |\eta|$ such that $\eta_i\models\phi_1$ and for
  all $0\leq j\leq i$ it holds $\eta_j\models\phi_2$ as well;
  \item If $\phi=\phi_1\wedge\phi_2$, then $\eta\models\phi$ iff
  $\eta\models\phi_1$ and $\eta\models\phi_2$;
  \item If $\phi=\phi_1\vee\phi_2$, then $\eta\models\phi$ iff
  $\eta\models\phi_1$ or $\eta\models\phi_2$.
\end{itemize}

The difference between the strong Next ($X$) and the weak Next ($X_w$)
operators is $X$ that requires a next state in the following while $X_w$
may not. Thus $X_w\phi$ is always true in the last state of a finite
trace, since no next state is provided.  As a result, in $\ltlf$ $X\ff$
is unsatisfiable, while $X_w\ff$ is satisfiable, which is quite different
with that in LTL, where neither $X\ff$ nor $\neg X \neg \ff$ are
satisfiable.

Let $\phi$ be an $\ltlf$ formula, we use $CF(\phi)$ to represent the
set of conjuncts in $\phi$, i.e.  $CF(\phi)=\{\phi_i | \phi_i\in
I\}$ if $\phi = \bigwedge_{i\in I}\phi_i$, where the root of
$\phi_i$ is not a conjunction.  $DF(\phi)$ (the set of disjuncts) is
defined analogously.

\subsection{The $\ltlf$ Satisfiability Problem}
The satisfiability problem is to check whether, for a given $\ltlf$
formua $\phi$, there is a finite trace $\eta\in\Sigma^*$ such that
$\eta\models\phi$:

\begin{definition}[$\ltlf$ Satisfiability Problem]
Given an $\ltlf$ formula $\phi$ over the alphabet $\Sigma$,
we say $\phi$ is satisfiable iff there is a finite trace
$\eta\in\Sigma^*$ such that $\eta\models\phi$.
\end{definition}

One approach is to reduce the $\ltlf$ satisfiability problem to
that of LTL.
\begin{theorem}[\cite{GV13}]
The Satisfiability problem for $\ltlf$ formulas is PSPACE-complete.
\end{theorem}
{\bf Proof Sketch}:
It is easy to reduce the $\ltlf$ satisfiability to LTL satisfiability:
\begin{enumerate}
\item Introduce a proposition ``\textit{Tail}'';
\item Require that \textit{Tail} holds at position 0;
\item Require also that \textit{Tail} stays $\tt$ until
it turns into $\ff$, and after that stays $\ff$ forever
($\textit{Tail} U (G\neg\textit{Tail})$).
\item The $\ltlf$ formula $\phi$ is translated into a
corresponding LTL formula in the following way:
\begin{itemize}
\item $t(p) \rightarrow p$, where $p$ is a literal;
\item $t(\neg\phi)=\neg t(\phi)$;
\item $t(\phi_1\wedge\phi_2)\rightarrow t(\phi_1)
\wedge t(\phi_2)$;
\item $t(\phi_1\vee\phi_2)\rightarrow t(\phi_1) \vee t(\phi_2)$;
\item $t(X\psi)\rightarrow X(\textit{Tail}\wedge t(\psi))$;
\item $t(\phi_1 U\phi_2)\rightarrow \phi_1 U (\textit{Tail}\wedge t(\phi_2))$;
\end{itemize}
\end{enumerate}
(The translation here does not require $\phi$ in NNF.
Thus the $X_w$ and $R$ operators can be handled by the rules
$X_w\phi\equiv \neg X\neg\phi$ and
$\phi_1 R\phi_2\equiv\neg(\neg \phi_1 U\neg\phi_2)$.)
Finally one can prove that $\phi$ is satisfiable iff
$\textit{Tail}\wedge \textit{Tail} U (G\neg\textit{Tail})\wedge t(\phi)$
is satisfiable. $\qed$

The reduction approach can take advantage of existing LTL satisfiability
solvers.  But, there may be an overhead as we need to find a
\emph{fair cycle} during LTL satisfiability checking, which is not
necessary in $\ltlf$ checking.

\subsection{$\ltlf$ Transition System}\label{sec:lts}
In \cite{LZPVH13}, Li et al. have proposed using transition systems for
checking satisfiability of LTL formulas. Here we adapt this approach
to $\ltlf$.  First, we define the \textit{normal form} for $\ltlf$
formulas.

\begin{definition}[Normal Form]\label{def:expansion}
The \textit{normal form} of an $\ltlf$ formula $\phi$, denoted as
$\expand(\phi)$, is a formula set defined as follows:
\begin{itemize}
\item $\expand(\phi) =\{\phi \wedge X(\tt)\}$ if $\phi\not\equiv\ff$ is a
  propositional formula. If $\phi\equiv\ff$, we define $\expand(\ff)=\emptyset$;
\item $\expand(X\phi/X_w\phi) = \{\tt\wedge X(\psi)\mid \psi\in DF(\phi)\}$;
\item $\expand(\phi_1 U \phi_2) = \expand(\phi_2)\cup \expand( \phi_1 \wedge
  X(\phi_1 U \phi_2))$;
\item $\expand(\phi_1 R \phi_2) = \expand(\phi_1 \wedge \phi_2) \cup \expand( \phi_2 \wedge X(\phi_1 R \phi_2))$;
\item $\expand(\phi_1 \vee \phi_2) = \expand(\phi_1)\cup \expand(\phi_2)$;
\item $\expand(\phi_1\wedge\phi_2) = \{(\alpha_1\wedge\alpha_2) \wedge X(\psi_1\wedge\psi_2)\mid \forall i=\mathit{1,2}. \ \alpha_i\wedge X(\psi_i)\in \expand(\phi_i)\}$;
\end{itemize}
For each $\alpha_i\wedge X\phi_i\in\dnf{\phi}$, we say it a
clause of $\dnf{\phi}$.
\end{definition}

(Although the normal forms of $X$ and $X_w$ formulas are the same,
we do distinguished bethween them through the accepting conditions
introduced below.)
Intuitively, each clause $\alpha_i\wedge X\phi_i$ of $\dnf{\phi}$
indicates that the propositionl formula $\alpha_i$ should hold now and
then $\phi_i$ should hold in the next state.  For $\phi_i$, we can also
compute its normal form. We can repeat this procedure until no new states
are required.
\begin{definition}[$\ltlf$ Transition System]\label{def:lts}
Let $\phi$ be the input formula. The
labeled transition system $T_{\phi}$ is a tuple $\langle Act,
S_\phi, \tran{}, \phi\rangle$ where:
1). $\phi$ is the initial state;
2). $Act$ is the set of conjunctive formulas over $L_{\phi}$;
3). the transition relation
$\tran{}\ \subseteq S_\phi\times Act\times S_\phi$ is defined by:
$\psi_1\tran{\alpha}\psi_2$ iff there exists $\alpha\wedge
X(\psi_2) \in \expand(\psi_1)$;
and 4). $S_\phi$ is the smallest set of formulas such that
$\psi_1\in S_\phi$, and $\psi_1\tran{\alpha}\psi_2$ implies
$\psi_2\in S_\phi$.

\end{definition}

Note that in LTL transition systems the $\ff$ state can be deleted,
as it can never be part of a fair cycle.  This state must be kept
in $\ltlf$ transition systems: a finite trace that reach $\ff$ may
be accepted in $\ltlf$, cf. $X_w\ff$.  Nevertheless, $\ff$ edges are
not allowed both in $\ltlf$ and LTL transition systems.

A \emph{run} of $T_{\phi}$ on finite trace
$\eta=\omega_0\omega_1\ldots\omega_n\in\Sigma^*$ is a sequence
$s_0\tran{\alpha_0}s_1\tran{\alpha_1}\ldots s_n\tran{\alpha_n}s_{n+1}$
such that $s_0=\phi$ and for every $0\leq i\leq n$ it holds
$\omega_i\models\alpha_i$. We say $\psi$ is \emph{reachable} from $\phi$ iff there is a
run of $T_{\phi}$ such that the final state is $\psi$.

\section{$\ltlf$ Satisfiability-Checking Framework}\label{sec:checking}
In this section we present our framework for checking satisfiability
of $\ltlf$ formulas.
First we show a simple lemma concerning finite
sequences of length $1$.
\begin{lemma}\label{lemma:last}
  For a finite trace $\eta\in\Sigma^*$ and $\ltlf$ formula
  $\phi$, if $|\eta|=1$ then $\eta\models\phi$ holds iff:
  \begin{itemize}
    \item $\eta\models\tt$ and $\eta\not\models\ff$;
    \item If $\phi=p$ is a literal, then return
    true if $\phi\in\eta$. otherwise return false;
    \item If $\phi=\phi_1\wedge\phi_2$, then return
    $\eta\models\phi_1$ and $\eta\models\phi_2$;
    \item If $\phi=\phi_1\vee\phi_2$, then return
    $\eta\models\phi_1$ or $\eta\models\phi_2$;
    \item If $\phi=X\phi_2$, then return false;
    \item If $\phi=X_w\phi_2$, then return true;
    \item If $\phi=\phi_1U\phi_2$ or $\phi=\phi_1R\phi_2$,
    then return $\eta\models\phi_2$.
  \end{itemize}
\end{lemma}
\begin{proof}
  This lemma can be directly proven from the semantics of $\ltlf$
  formulas by fixing $|\eta| = 1$.
\end{proof}

Now we characterize the satisfaction relation for finite sequences:
\begin{lemma}\label{lemma:lts}
  For a finite trace
  $\eta=\omega_0\omega_1\ldots\omega_n\in\Sigma^*$ and
  $\ltlf$ formula $\phi$,
  \begin{enumerate}
  \item If $n=0$, then $\eta\models\phi$ iff there exists
    $\alpha_i\wedge X\phi_i\in\dnf{\phi}$ such that
    $\omega_0\models\alpha_i$ and $CF(\alpha_i)\models\phi$;
    \item If $n\ge 1$, then $\eta\models\phi$ iff there exists
    $\alpha_i\wedge X\phi_i\in\dnf{\phi}$ such that $\omega_0\models\alpha_i$ and
    $\eta_1\models\phi_i$;
    \item $\eta\models\phi$ iff there exists a run
    $\phi=\phi_0\tran{\alpha_0}\phi_1\tran{\alpha_1}\phi_2\ldots\tran{\alpha_n}\phi_{n+1}$
    in $T_{\phi}$ such that for every $0\leq i\leq n$ it holds
    that $\omega_i\models\alpha_i$ and $\eta_i\models\phi_i$.
  \end{enumerate}
\end{lemma}
\begin{proof}
\begin{enumerate}
\item $CF(\alpha_i)$ is treated to be a finite trace whose length is 1. We prove the first item by structural induction over
  $\phi$.  
  \begin{itemize}
  \item If $\phi=p$, then $\eta\models\phi$ iff $\omega_0\models p$ and $CF(p)\models\phi$ hold, where
  $p\wedge X\tt$ is actually in $\dnf{\phi}$;
  \item If
  $\phi=\phi_1\wedge\phi_2$, then $\eta\models\phi$ holds iff
  $\eta\models\phi_1$ and $\eta\models\phi_2$ hold, and iff by induction
  hypothesis, there exists $\beta_i\wedge X\psi_i$ in $\dnf{\phi_i}$
  such that $\omega_0\models\beta_i$ and $CF(\beta_i)\models\phi$
  ($i=1,2$). Let $\alpha_i=\beta_1\wedge\beta_2$ and
  $\phi_i=\psi_1\wedge\psi_2$, then according to Definition
  \ref{def:expansion} we know $\alpha_i\wedge X\phi_i$ is in
  $\dnf{\phi}$, and $\omega_0\models\alpha_i$ and
  $CF(\alpha_i)\models\phi$ hold; The proof for the case when $\phi=\phi_1\vee\phi_2$ is similar;
  \item Note that $\eta\models X\psi$ is always false,
  and if $\phi=X_w\psi$ then from Lemma \ref{lemma:last} it is always true that $\eta\models X_w\psi$ iff $\tt\wedge X\psi\in \dnf{\phi}$ and $\tt\models X_w\psi$;
  \item If
  $\phi=\phi_1U\phi_2$, then $\eta\models\phi$ holds iff
  $\eta\models\phi_2$ holds from Lemma \ref{lemma:last}, and iff by induction hypothesis, there exists
  $\alpha_i\wedge X\phi_i\in\dnf{\phi_2}$ such that
  $\omega_0\models\alpha_i$ and $CF(\alpha_i)\models\phi_2$, and thus $CF(\alpha_i)\models\phi$
  according to $\ltlf$ semantics.  From
  Definition \ref{def:expansion} we know as well that
  $\alpha_i\wedge X\phi_i$ is in $\dnf{\phi}$, thus the proof is done;
  The proof for the case when $\phi=\phi_1R\phi_2$ is similar;
  \end{itemize}

  \item The second item is also proven by structural induction over $\phi$.
  \begin{itemize}
  \item If $\phi=\tt$ or $\phi=p$, then $\eta\models\phi$ iff $\omega_0\models\phi$ and $\eta_1\models\tt$ hold, where
  $\phi\wedge X\tt$ is actually in $\dnf{\phi}$;
  \item If $\phi=X\phi_2$ or $\phi=X_w\phi_2$, since $|\eta|>1$ so it is obviously true that $\eta\models\phi$ iff
  $\omega_0\models\tt$ and
  $\eta_1\models\phi_2$ hold according to $\ltlf$ semantics, and obviously $\tt\wedge X\phi_2$ is
  in $\dnf{\phi}$;
  \item If $\phi=\phi_1\wedge\phi_2$, then $\eta\models\phi$ iff $\eta\models\phi_1$ and $\eta\models\phi_2$, and iff
  by induction hypothesis, there exists $\beta_i\wedge X\psi_i\in \dnf{\phi_i}(i=1,2)$ such that
  $\omega_0\models\beta_i$ and $\eta_1\models\psi_i$ hold, and iff $\omega_0\models\beta_1\wedge\beta_2$
  and $\eta_1\models\psi_1\wedge\psi_2$ hold, in which $(\beta_1\wedge\beta_2)\wedge X(\psi_1\wedge\psi_2)$ is
  indeed in $\dnf{\phi}$; The case when $\phi=\phi_1\vee\phi_2$ is similar;
  \item If $\phi=\phi_1 U\phi_2$, then $\eta\models\phi$ iff $\eta\models\phi_2$ or $\eta\models(\phi_1\wedge X\phi)$.
  If $\eta\models\phi_2$ holds, then by induction hypothesis iff there exists $\alpha_i\wedge X\phi_i\in\dnf{\phi_2}$ such that
  $\omega_0\models\alpha_i$ and $\eta_1\models\phi_i$. According to Definition \ref{def:expansion} we know $\alpha_i\wedge X\phi_i$ is also $\dnf{\phi_2}$. On the other hand, if $\eta\models\phi_1\wedge X\phi$ holds, the proofs for $\wedge$ formulas are already done. Thus,
  it is true that $\eta\models\phi$ iff there exists $\alpha_i\wedge X\phi_i\in\dnf{\phi_2}$ such that
  $\omega_0\models\alpha_i$ and $\eta_1\models\phi_i$; The case when $\phi=\phi_1 R\phi_2$ is similar to prove.
  \end{itemize}

  \item Applying the first item if $n=0$ and recursively applying the second item if $n\geq 1$, we can prove the third item.
\end{enumerate}
\end{proof}

Lemma \ref{lemma:lts} states that, to check whether a finite trace
$\eta=\omega_0\omega_1\ldots\omega_n$ satisfies the $\ltlf$
formula $\phi$, we can find a run of $T_{\phi}$ on $\eta$ such that
$\eta$ can finally reach the transition
$\phi_n\tran{\alpha_n}\phi_{n+1}$ and satisfies
$\omega_n\models\alpha_n$, and moreover $CF(\alpha_n)\models\phi_{n}$.
Now we can give the main theorem of this paper.

\begin{theorem}\label{theorem:main}
  Given an $\ltlf$ formula $\phi$ and a finite trace
  $\eta=\omega_0\ldots\omega_n(n\geq 0)$, we have that
  $\eta\models\phi$ holds iff there exists a run of $T_\phi$ on $\eta$ which
  ends at the transition $\psi_1\tran{\alpha}\psi_2$ satisfying
  $CF(\alpha)\models\psi_1$.
\end{theorem}
\begin{proof}
   Combine the first and third items in Lemma \ref{lemma:lts}, and
   we can easily prove this theorem.
\end{proof}

We say the state $\psi_1$ in $T_{\phi}$ is \emph{accepting}, if there
exists a transition $\psi_1\tran{\alpha}\psi_2$ such that
$CF(\alpha)\models\psi_1$.  Theorem
\ref{theorem:main} implies that, the formula $\phi$ is satisfiable if
and only if there exists an accepting state $\psi_1$ in $T_{\phi}$
which is reachable from the initial state $\phi$. Based on this
observation, we now propose a simple on-the-fly satisfiability-checking framework for $\ltlf$ as follows:
  \begin{enumerate}
    \item If $\phi$ equals $\tt$, return $\phi$ is \textit{satisfiable};
    \item The checking is processed on the transition system $T_{\phi}$ on-the-fly, i.e. computing the
    reachable states step by step with the DFS (Depth First Search) manner, until an accepting one
    is reached: Here we return \textit{satisfiable};
    \item Finally we return \textit{unsatisfiable} if all states in the
    whole transition system are explored.
  \end{enumerate}

  The complexity of our algorithm mainly depends on the size of
  constructed transition system. The system construction is
  the same as the one for LTL proposed in \cite{LZPVH13}.  Given an
  $\ltlf$ formula $\phi$, the constructed transition system $T_{\phi}$
  has at worst the size of $2^{cl(\phi)}$, where $cl(\phi)$ is the set
  of subformulas of $\phi$.

  \section{Optimizations}\label{sec:optimizations}

 In this section we propose some optimization strategies by exploiting SAT
 solvers. First we study the relationship between the satisfiability problems
 for $\ltlf$ and LTL formulas.

\subsection{Relating to LTL Satisfiability}

In this section we discuss some connections between $\ltlf$ and LTL
formulas. We say an $\ltlf$ formula $\phi$ is \emph{$X_w$-free} iff $\phi$
does not have the $X_w$ operator. Note that$\ltlf$ formulas may
contain the $X_w$ operator, while standard LTL ones do not. Here consider
$X_w$-free formulas, in which $\ltlf$ and LTL have the same syntax.
First the following lemma shows how to extend a finite trace into an
infinite one but still preserve the satisfaction from $\ltlf$ to LTL:

\begin{lemma}\label{lemma:f2inf}
  Let $\eta = \omega_0$ and $\phi$ an $\ltlf$ formula which is $X_w$-free,
  then $\eta\models\phi$ implies $\eta^{\omega}\models\phi$ when $\phi$ is
  considered as an LTL formula.
\end{lemma}

\begin{proof}
  We prove it by structural induction over $\phi$:
  \begin{itemize}
    \item If $\phi$ is a literal $p$, then $\eta\models p$ implies $p\in\eta$. Thus $\eta^{\omega}\models\phi$ is true; And if $\phi$ is $\tt$, then $\eta^{\omega}\models\tt$ is obviously true;
    \item If $\phi=\phi_1\wedge\phi_2$, then $\eta\models\phi$ implies $\eta\models\phi_1$ and $\eta\models\phi_2$. By induction
    hypothesis we have $\eta^{\omega}\models\phi_1$ and $\eta^{\omega}\models\phi_2$. So $\eta^{\omega}\models\phi_1\wedge\phi_2$; The proof is similar when $\phi=\phi_1\vee\phi_2$;
    \item If $\phi=X\psi$, then according to Lemma \ref{lemma:last}
    we know $\eta\models\phi$ cannot happen; And since $\phi$ is $X_w$-free, so $\phi$ cannot be a $X_w$ formula;

    \item If $\phi=\phi_1U\phi_2$, then $\eta\models\phi$ implies
    $\eta\models\phi_2$ according to Lemma \ref{lemma:last}. By induction hypothesis we have $\eta^{\omega}\models\phi_2$. Thus
    $\eta^{\omega}\models\phi$ is true from the LTL semantics;
    Similarly when $\phi=\phi_1R\phi_2$, we know for every $i\geq 0$ it is true that $(\xi_i=\eta^{\omega})\models\phi_2$. Thus
    $\eta^{\omega}\models\phi$ holds from the LTL semantics; The proof is done.

  \end{itemize}
\end{proof}

We showed earlier that $\ltlf$ satisfiability can be reduced to LTL
satisfiability problem.  We show that the satisfiability of some
$\ltlf$ formulas implies satisfiability of LTL formulas:
\begin{theorem}\label{thm:f2infsat}
Let $\phi$ be an $X_w$-free formula.
If $\phi$ is satisfiable as an $\ltlf$ formula,
then $\phi$ is also satisfiable as an LTL formula.
\end{theorem}

\begin{proof}
Assume $\phi$ is a $X_w$-free $\ltlf$ formula, and is satisfiable. Let
$\eta=\omega_0\ldots\omega_n$ such that $\eta\models\phi$. Now we
interpret $\phi$ as an LTL formula.  Combining Lemma
\ref{lemma:lts} and Lemma \ref{lemma:f2inf}, we get that
$\xi\models\phi$ where $\xi=\omega_0\ldots\omega_{n-1}(\omega_n)^\omega$.
 \end{proof}
Equivalently, if $\phi$ is an LTL formula and $\phi$ is
unsatisfiable, then the $\ltlf$ formula $\phi$ is also unsatisfiable.
Note here the $\ltlf$ formula $\phi$ is $X_w$-free since it can be
considered as an LTL formula.

\begin{figure}[t]
\begin{minipage}[b]{0.45\linewidth}
\centering
\scalebox{0.8}{
\begin{tikzpicture}[>=stealth,shorten >=1pt,node distance=2cm,on grid,auto]
   \node[state,initial] (q_0)   {$\phi$}

   (q_0) edge [loop above] node {$\tt$} (q_0)
   (q_0) edge [loop below] node {$a$} (q_0)
   (q_0) edge [loop right] node {$\neg a$} (q_0);
\end{tikzpicture}
}
\caption{The transition system of $\phi=GF a \wedge GF \neg a$.}
\label{fig:lts1}
\end{minipage}
\hspace{0.3cm}
\begin{minipage}[b]{0.45\linewidth}
\centering
\scalebox{0.8}{
\begin{tikzpicture}[>=stealth,shorten >=1pt,node distance=2cm,on grid,auto]
   \node[state,initial] (q_0)   {$\phi$};
   \node[state] (q_1) [right=of q_0] {$\phi_1$};
   \path[->]
   (q_0) edge [loop above] node {$b$} (q_0)
   (q_0) edge [bend left] node [above] {$a$} (q_1)
   (q_1) edge node {$b$} (q_0)
   (q_1) edge [loop right] node {$a$} (q_1);
\end{tikzpicture}
}
\caption{The transition system of $\phi=G(a U b)$. Note $\phi_1=\phi\wedge a U b$}
\label{fig:lts2}

\end{minipage}
\end{figure}

\begin{example}
\begin{itemize}
\item Consider the $X_w$-free formula $\phi=GF a \wedge GF\neg a$, whose
  transition system is shown in Figure \ref{fig:lts1}. If $\phi$ is
  treated as an LTL formula, then we know that the infinite trace
  $(\{a\}\{\neg a\})^{\omega}$ satisfies $\phi$. However, if $\phi$ is
  considered to be an $\ltlf$ formula, then we know from that no
  accepting state exists in the transition system, so it is
  unsatisfiable. It is due to the fact that no transition
  $\psi_1\tran{\alpha}\psi_2$ in $T_{\phi}$ satisfies the condition
  $CF(\alpha)\models \psi_1$.
\item Consider another example formula $\phi=G(a U b)$, whose
  transition system is shown in Figure \ref{fig:lts2}. Here we can
  find an accepting state ($\phi$, as $\phi\tran{b}\phi$ and $CF(b)\models\phi$ hold). Thus we know that $\phi$ is
  satisfiable, interpreted over both finite or infinite traces.
\end{itemize}
\end{example}

\subsection{Obligation Formulas}
For an LTL formula $\phi$, Li et al. \cite{LZPVH13} have defined its \textit{obligation
  formula} $\of{\phi}$ and show that if
$\of{\phi}$ is satisfiable then $\phi$ is satisfiable. Since $\of{\phi}$ is essentially a boolean
formula, so we can check it efficiently using modern SAT solvers. However this
cannot apply to $\ltlf$ directly, which we illustrate in the
following example.

\begin{example}
Consider $\phi=GXa$, where $\alpha$ is a satisfiable propositional
formula. It is easy to see that it is satisfiable if it
is an LTL formula (with respect to some word $a^\omega$), while
unsatisfiable when it is an $\ltlf$ formula (because no finite trace
can end with the point satisfying $Xa$).
From \cite{LPZHVCoRR14}, the obligation formula of $\phi$ is
$\of{\phi}=a$, which is obviously satisfiable. So the satisfiability
of obligation formula implies the satisfiability of LTL formulas,
but not that of $\ltlf$ formulas.
\end{example}

We now show how to handle of Next operators ($X$ and $X_w$) after the
Release operators.  For a formula $\phi$, we define three obligation
formulas:
\begin{definition}[Obligation Formulas]\label{def:ofs}
Given an $LTL_f$ formula $\phi$, we define three kinds of obligation
formulas: global obligation formula, release obligation
formula, and general obligation formula--denoted
as $\ofg{\phi}$, $\ofr{\phi}$ and $\off{\phi}$, by induction over $\phi$.
(We use $\textit{ofx}$ as a generic reference to $\textit{ofg}$, $\textit{ofr}$, and $\textit{off}$.)
\begin{itemize}
\item
$\ofi{\phi}=\tt$ if $\phi=\tt$; and $\ofi{\phi}=\ff$ if $\phi=\ff$;
\item
If $\phi=p$ is a literal, then $\ofi{\phi}=p$;
\item
If $\phi=\phi_1\wedge\phi_2$,
then $\ofi{\phi}= \ofi{\phi_1}\wedge \ofi{\phi_2}$;
\item
If $\phi=\phi_1\vee\phi_2$, then $\ofi{\phi}= \ofi{\phi_1}\vee \ofi{\phi_2}$;
\item
If $\phi=X\phi_2$, then $\off{\phi} = \off{\phi_2}$, $\ofr{\phi}=\ff$ and $\ofg{\phi}=\ff$;
\item
If $\phi=X_w\phi_2$, then $\off{\phi} = \off{\phi_2}$, $\ofr{\phi}=\ff$ and $\ofg{\phi}=\tt$;
\item If $\phi=\phi_1 U\phi_2$, then $\ofi{\phi}=\ofi{\phi_2}$.
\item If $\phi=\phi_1R\phi_2$, then $\off{\phi}=\ofr{\phi}$, $\ofr{\phi}=\ofr{\phi_2}$ and $\ofg{\phi}=\ofg{\phi_2}$
\end{itemize}
For example in the third item, the equation represents actually three:
$\off{\phi}= \off{\phi_1}\wedge \off{\phi_2}$, $\ofr{\phi}= \ofr{\phi_1}\wedge \ofr{\phi_2}$ and $\ofg{\phi}= \ofg{\phi_1}\wedge \ofg{\phi_2}$.
\end{definition}

For $\off{\phi}$, the changes in comparison to \cite{LPZHVCoRR14} are the
definition for release formulas, and introducing the $X_w$ operator.
For example, we have that $\off{GXa}$ is $\ff$ rather than $a$. Moreover,
since the $\ltlf$ formula $GX_wa$ is satisfiable, the definition of
$\ofg{\phi}$ is required to identify this situation. (Below we show
a fast satisfiability-checking strategy that uses global obligation
formulas.)

The obligation-acceleration optimization works as follows:
\begin{theorem}[Obligation Acceleration]\label{thm:off}
For an $\ltlf$ formula $\phi$, if $\off{\phi}$ is satisfiable then
$\phi$ is satisfiable.
\end{theorem}
\begin{proof}
Since $\off{\phi}$ is satisfiable, there exists $A\in\Sigma$ such
that $A\models\off{\phi}$.  We prove that there exists $\eta=A^n$
where $n\geq 1$ such that $\eta\models\phi$, by structural
induction over $\phi$. Note the cases $\phi=\tt$ or $\phi=p$ are
    trivial. For other cases:
    \begin{itemize}
    \item If $\phi=\phi_1\wedge\phi_2$, then
    $\off{\phi}=\off{\phi_1}\wedge\off{\phi_2}$ from Definition
    \ref{def:ofs}.  So $\off{\phi}$ is satisfiable implies that there
    exists $A\models\off{\phi_1}$ and $A\models\off{\phi_2}$. By
    induction hypothesis there exists $\eta_i=A^{n_i}$ ($n_i\ge 0$)
    such that $\eta_i\models\phi_i$ ($i=1,2$). Assume $n_1\geq n_2$,
    then let $\eta=\eta_1$. Then, $\eta\models\phi_1\wedge\phi_2$.
    \item If $\phi=X\phi_2$ or $\phi=X_w\phi_2$, then $\off{\phi}$ is satisfiable
    iff $\off{\phi_2}$ is satisfiable. So there exists $A$ models $\phi_2$.
    By induction hypothesis, there exists $n$ such that $A^n\models\phi_2$, thus
    according to $\ltlf$ semantics, we know $A^{n+1}\models\phi$;
    \item If$\phi=\phi_1 R\phi_2$, then $\off{\phi}=\ofr{\phi_2}$. Thus
    $\ofr{\phi_2}$ is also satisfiable. So there exists $A\models\ofr{\phi_2}$, based on which we can
    show that $A\models\phi_2$ by structural induction over
    $\phi_2$ by a similar proof. Thus Let $\eta=A$ and according to Lemma
    \ref{lemma:last} we know $\eta\models\phi_2$ implies
    $\eta\models\phi$. The case for Until can be treated in a similar way, thus the proof is done.
    \end{itemize}
  \end{proof}

\subsection{A Complete Acceleration Technique for Global Formulas}
The obligation-acceleration technique (Theorem \ref{thm:off}) is sound
but not complete, see the formula $\phi=a \wedge GF(\neg a)$, in which
$\off{\phi}$ is unsatisfiable, while $\phi$ is, in fact, satisfiable.
In the following, we prove that both soundness and completeness hold for
the \emph{global} $\ltlf$ formulas, which are formulas of the form of
$G\psi$, where $\psi$ is an arbitrary $\ltlf$ formula.

\begin{theorem}[Obligation Acceleration for Global formulas]\label{thm:global}
For a global $\ltlf$ formula $\phi=G\psi$, we have that $\phi$ is
satisfiable iff $\ofg{\psi}$ is satisfiable.
  \end{theorem}
  \begin{proof}
    For the forward direction, assume that $\phi$ is satisfiable.
    It implies that there is a
    finite trace $\eta$ satisfying $\phi$. According to Theorem
    \ref{theorem:main}, $\eta$ can run on $T_{\phi}$ and reaches an
    accepting state $\psi_1$, i.e., $\psi_1\tran{\alpha}\psi_2$ and
    $CF(\alpha)\models\psi_1$. Since $\phi$ is a global formula and
    $\psi_1$ is reachable from $\phi$, it is not hard to prove that
    $CF(\phi)\subseteq CF(\psi_1)$ from Definition \ref{def:lts}. So
    $CF(\alpha)\models\phi$ is also true. Since $\phi$ is a global
    formula so $CF(\alpha)\models\psi$ holds from Lemma
    \ref{lemma:last}. Then one can prove that
    $CF(\alpha)\models\ofg{\psi}$ by structural induction over $\psi$
    (it is left to readers here),
    which implies that $\ofg{\psi}$ is satisfiable.

    For the backward direction, assume $\ofg{\psi}$ is satisfiable. So
    there exists $A\in\Sigma$ such that $A\models\ofg{\psi}$. Then one
    can prove $A\models\phi$ is also true by structural induction over
    $\psi$ ($\phi=G\psi$). For paper limit, this proof is left to readers.
    So $\phi$ is satisfiable. The proof is done.
  \end{proof}

\subsection{Acceleration for Unsatisfiable Formulas}\label{sec:unsatacc}
Theorem \ref{thm:f2infsat} indicates that if an LTL formula $\phi$
(of course $X_w$-free) is unsatisfiable, then the $LTL_f$ formula
$\phi$ is also unsatisfiable.  As a result, optimizations for
unsatisfiable LTL formulas, for instance those  in \cite{LPZHVCoRR14},
can be used directly to check unsatisfiable $X_w$-free $\ltlf$ formulas.

\section{Experiments}\label{sec:exp}
In this section we present an experimental evaluation. The algorithms are
implemented in the \tool\ tool%
\footnote{Tool will be released upon paper publication.}.
We have implemented three optimization strategies. They are
1). \textit{off}: the obligation acceleration technique for $\ltlf$(Theorem \ref{thm:off});
2). \textit{ofg}: the obligation acceleration for global $\ltlf$
  formula (Theorem \ref{thm:global}); 3). \textit{ofp}: the acceleration for unsatisfiable formulas
  (Section \ref{sec:unsatacc}).
Note that all three optimizations can benefit from the power of modern
SAT solvers.

We compare our algorithm with the approach using off-the-shelf tools
for checking LTL satisfiability.  We choose the tool \polsat,
a portfolio LTL solver, which was introduced in
\cite{LPZVHCoRR13}. One main feature of \polsat\ is that it integrates
most existing LTL satisfiability solvers; consequently, it is
currently the best-of-breed LTL satisfiability solver.
The input of \tool\ is directly an $\ltlf$ formula $\phi$,
while that of \polsat\ should be $\textit{Tail}\wedge \textit{Tail} U
G(\neg\textit{Tail})\wedge t(\phi)$, which is the LTL formula that is
equi-satisfiable with the $\ltlf$ formula $\phi$.

The experimental platform of this paper is a cluster
that consists of 47 IBM Power 755 nodes,
each of which contains four eight-core POWER7 processors running
at 3.86GHz.
In our experiments, both \tool\ and \polsat\ occupy a unique node,
and \polsat\ runs all its integrated solvers in parallel by using
independent cores of the node. The time is measured by Unix time
command, and each test case has the maximal limitation of 60
seconds.

Since LTL formulas are also $\ltlf$ formulas, we use
existing LTL benchmarks to test the tools. We compare the results
from both tools, and no inconsistency occurs.

\subsection{Schuppan-collected Formulas}
We consider first the benchmarks introduced in previous works \cite{SD11}.
The benchmark suite there include earlier benchmark suites (e.g.,
\cite{RV10}), and we refer to this suite as \textit{Schuppan-collected}.
The \textit{Schuppan-collected} suite has a total amount of 7448 formulas.
The different types of benchmarks are shown in
the first column of Table \ref{tab:schuppan}.

\begin{table}
\centering
    \caption{Experimental results on Schuppan-collected formulas.}\label{tab:schuppan}
    \scalebox{0.8}
    {
    \begin{tabular}{|l|r|r|r|}
    \hline
    Formula type  &  \tool (sec.)  &  \polsat (sec.) & \polsat /\tool\\
    \hline
      /acacia/example  &  1.5  &  3.3 & 2.2\\
\hline
/acacia/demo-v3  &  1.4  &  604.7 & 431.9\\
\hline
/acacia/demo-v22  &  2.0  &  1.3 & 0.65\\
\hline
/alaska/lift  &  23.0  &  7319.6 & 318.2\\
\hline
/alaska/szymanski  &  1.2  &  7.3 & 6.1\\
\hline
/anzu/amba  &  2120.9 &  2052.9 & 0.97\\
\hline
/anzu/genbuf  &  3606.9  &  3717.9 & 1.0\\
\hline
/rozier/counter  &  1840.3  &  3009.3 & 1.6\\
\hline
/rozier/formulas  &  552.9  &  467.0 & 0.8\\
\hline
/rozier/pattern  &  22.9  &  49.9 & 2.1\\
\hline
/schuppan/O1formula  &  2.9  &  7.1 & 2.4\\
\hline
/schuppan/O2formula  &  3.1  &  1265.0 & 408.1\\
\hline
/schuppan/phltl  &  226.3  &  602.5 & 2.6\\
\hline
/trp/N5x  &  10.5  &  42.0 & 4.0\\
\hline
/trp/N5y  &  2764.9  &  2777.4 & 1.0\\
\hline
/trp/N12x  &  22.8  &  24061.1 & 1055.3\\
\hline
/trp/N12y  &  4040.2  &  4049.2 & 1.0\\
\hline
Total  &  15244.2  &  50038.2 & 3.2\\
\hline
\end{tabular}
}
\end{table}

Table \ref{tab:schuppan} shows the experimental results on
\textit{Schuppan-collected} benchmarks. The fourth column of the table
shows the speed-up from \tool\ to \polsat. One can see that the results
from \tool\ outperforms those from \polsat, often by several orders of
magnitudes. We explain some of them.

The formulas in ``Schuppan-collected/alaska/lift'' are mostly
  unsatisfiable, which can be handled by the \textit{ofg} technique of
  \tool. On the other side, \polsat\ needs more than 300 times to
  finish the checking. The same happens on the
  ``Schuppan-collected/trp/N12x'' patterns, in which \tool\ is more
  than 1000 times faster. For the ``Schuppan-collected/schuppan/O2formula'' pattern
  formulas, \tool\ scales better due to the \textit{ofp} technique.

  Among the results from \tool, totally 5879 out of 7448 formulas
  in the benchmark are checked by using the \textit{off}
  technique. This indicates the \textit{off} technique is very efficient.
  Moreover, 84 of them are finished by exploring whole
  system in the worst time, which requires further improvement.
Overall, we can see \polsat\ is three times slower on this benchmark
suite than \tool.

\begin{figure}[t]
\begin{minipage}[b]{0.45\linewidth}
\centering
\includegraphics[scale = 0.53]{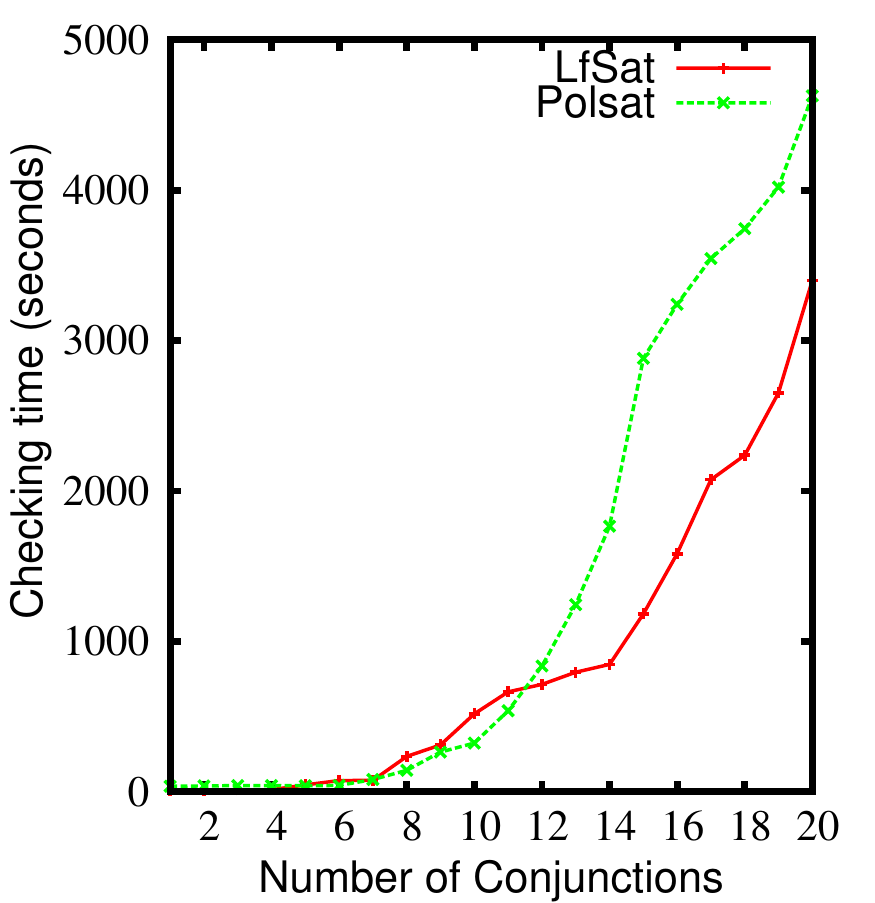}
      \caption{Experimental results on random conjunction formulas.}\label{fig:rc}
\vspace{0.3cm}
\end{minipage}
\hspace{0.3cm}
\begin{minipage}[b]{0.45\linewidth}
\centering
\includegraphics[scale = 0.53]{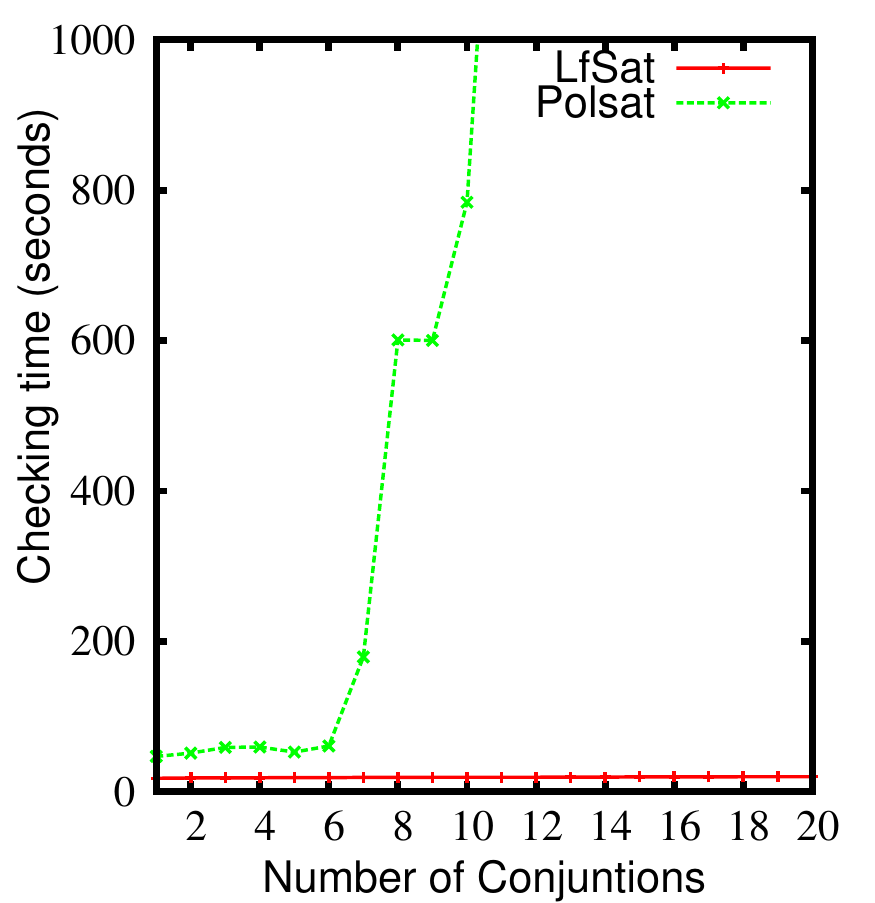}
\caption{Experimental results on global random conjunction formulas.}\label{fig:rc_global}
\end{minipage}
\end{figure}

\subsection{Random Conjunction Formulas}
Random conjunction formulas have the form of $\bigwedge_{1\leq i\leq
  n} P_i$, where $P_i$ is randomly selected from typical small
pattern formulas widely used in model checking \cite{LZPVH13}.
By randomly choosing the that atoms the small patterns use, a large number
of random conjunction formulas can be generated. More specially, to
evaluate the performance on global formulas, we also fixed the selected
$P_i$ by a random global pattern, and thus create a set of global
formulas.  In our experiments, we test 10,000 cases each for both random
conjunction and global random conjunction formulas, with the number of
conjunctions varying from 1 to 20 and 500 cases for each number.

Figure \ref{fig:rc} shows the comparison results on random conjunction
formulas. On average \tool\ earns about 10\% improving performance on
this kind of formulas. Among all the 10,000 cases, 8059 of
them are checked by the \textit{off} technique; 1105 of them are
obtained by the \textit{ofg} technique; 508 are acquired
by the \textit{ofp} technique; and another 107 are from an accepting state.
There are also 109 formulas equivalent to $\tt$ or $\ff$, which can be
directly checked.  In the worst case, 76 formulas are
finished by exploring the whole transition system. About 36 formulas
fail to be checked within 60 seconds by \tool. Statistics above show the optimizations
are very useful.

Moreover, one can conclude from Figure \ref{fig:rc_global} that,
\tool\ dominates \polsat\ when performing on the global random
conjunction formulas. As the \textit{ofg} technique is both sound and
complete for global formulas and invokes SAT solvers only once, so \tool\ performs almost constant time
for checking both satisfiable and unsatisfiable formulas.  Compared
with that, \polsat\ takes an ordinary checking performance for this
kind of special formulas. Indeed, the \textit{ofg} technique is
considered to play the crucial role on checking global $\ltlf$ formulas.

\section{Conclusion}\label{sec:con}
In this paper we have proposed a novel $\ltlf$ satisfiability-checking
framework based on the \textit{$\ltlf$ transition system}.  Meanwhile,
three different optimizations are introduced to accelerate the
checking process by using the power of modern SAT solvers, in which
particularly the \textit{ofg} optimization plays the crucial role on checking
global formulas.  The experimental results show that, the checking
approach proposed in this paper is clearly superior to the
reduction to LTL satisfiability checking.

\bibliography{ok,cav}

\end{document}